\documentclass[11pt,a4paper]{article}
\usepackage[T1]{fontenc}
\usepackage[utf8]{inputenc}
\usepackage{lmodern}
\usepackage[margin=28mm]{geometry}
\usepackage{amsmath,amssymb,amsthm,mathtools}
\usepackage{microtype,booktabs,enumitem}
\usepackage[hidelinks]{hyperref}
\usepackage[capitalize,nameinlink,noabbrev]{cleveref}
\usepackage{needspace}
\usepackage{placeins}
\setlist[enumerate]{leftmargin=*,itemsep=3pt,topsep=5pt}
\setlist[itemize]{leftmargin=*,itemsep=3pt,topsep=5pt}
\newtheorem{theorem}{Theorem}
\newtheorem{lemma}[theorem]{Lemma}
\newtheorem{proposition}[theorem]{Proposition}

\theoremstyle{remark}

\newcommand{\Greedy}{\textnormal{\textsc{Greedy}}}
\newcommand{\RePair}{\textnormal{\textsc{RePair}}}
\newcommand{\val}{\operatorname{val}}

\newcommand{\Fac}{\operatorname{Fac}}

\newcommand{\proofheading}[1]{\par\smallskip\noindent\textit{#1.}\ }
\title{A Non-constant Lower Bound for\newline Grammar-Based Compression with Greedy}
\date{}
\hypersetup{pdftitle={A Non-constant Lower Bound for Grammar-Based Compression with Greedy},pdfsubject={Grammar-based compression; approximation lower bounds}}
\allowdisplaybreaks[1]
\author{Danny Hucke}
\hypersetup{pdfauthor={Danny Hucke}}

\begin{document}
\maketitle

\begin{abstract}
We prove a lower bound of $\Omega(\log n/\log\log n)$ on the approximation ratio of the global grammar-based compression algorithm \Greedy{}. To our knowledge, the previously best lower bound was a constant, and the existence of a nonconstant lower bound had remained open for more than twenty years. Our bound holds on an infinite family of words of length $n$, over alphabets of growing size, for every execution using left-to-right occurrence replacement and arbitrary tie-breaking. The lower bound is also formally verified in Lean~4.
\end{abstract}

\section{Introduction}

The smallest grammar problem asks for a smallest context-free grammar generating a given word and no other word. We measure a grammar by the total length of its right-hand sides. This problem and several natural approximation algorithms were studied systematically by Charikar et al.~\cite{charikar}. One of these algorithms is \Greedy{}, based on the greedy textual substitution methods of Apostolico and Lonardi~\cite{apostolico-lonardi-1998,apostolico-lonardi-2000}. In the global formulation analysed by Charikar et al., each round replaces a repeated factor throughout the current grammar, choosing a factor that gives the largest immediate reduction in grammar size.

Charikar et al.~\cite{charikar} established a constant lower bound for \Greedy{} in 2005. Hucke and Reh~\cite{hucke-reh} later improved this bound to approximately $1.34847$, using unary inputs. To the best of our knowledge, obtaining a nonconstant lower bound for \Greedy{} has remained open since the work of Charikar et al., more than twenty years ago.

Our lower bound reaches a scale already highlighted by Charikar et al.~\cite{charikar} through their connection between grammar compression and addition chains. Their reduction shows that a polynomial-time grammar compressor with approximation ratio $o(\log n/\log\log n)$ would improve on Yao's method~\cite{yao} for constructing an addition chain containing a prescribed set of integers. This suggested that their original lower bounds for global compressors might substantially underestimate their worst-case behaviour. Bannai et al.~\cite{bannai} subsequently established an $\Omega(\log n/\log\log n)$ lower bound for \RePair{}. Here, we establish the same lower bound for \Greedy{}.

Lower-bound constructions for grammar compressors must account for the substitutions that the algorithm actually performs. For \Greedy{}, this includes substitutions in rules created during earlier rounds. We address this difficulty by forcing a sequence of substitutions while allowing arbitrary intervening compression of already processed material. The target words encode a de Bruijn word, following a general approach also used in the analysis of \RePair{} by Bannai et al.~\cite{bannai}. Auxiliary words based on a power-free morphism make the required substitutions preferable to every competing substitution that could change a target.

\Needspace{17\baselineskip}
Throughout the paper, logarithms are to base two. Our main result is the following.

\begin{theorem}\label{thm:main}
For every $h=2^r$ with $r\geq 10$, there is an explicitly specified word $w_h$ over an alphabet of size $h^2+4h$ such that, writing $n_h=|w_h|$,
\[
 h^2+4h\leq g(w_h)\leq 600h^2,
 \qquad 2^h\leq n_h\leq 2^{50h}.
\]
Every terminal grammar $G_h$ produced by global \Greedy{} on $w_h$, using left-to-right occurrence replacement and any choice among maximizing factors, satisfies
\[
 |G_h|\geq \frac{h^3}{8(4r+1)}.
\]
Consequently,
\[
 \frac{|G_h|}{g(w_h)}
 \geq \frac{h}{4800(4r+1)}
 \geq \frac{\log n_h}{1200000\log\log n_h}.
\]
\end{theorem}

\paragraph{Formal verification.}
The lower bound has also been formalised in Lean~4. The formalisation covers the construction and the analysis of all complete executions of \Greedy{}, including the finite morphism calculations. The proof sources and instructions for reproducing the verification are provided in~\cite{lean-proof}; \cref{app:lean} summarises their scope.

\section{Preliminaries}\label{sec:preliminaries}

A \emph{straight-line program} (SLP) over a finite alphabet $\Sigma$ consists of finitely many nonterminals, a designated start nonterminal $S$, and one production $A\to\alpha_A$ for each nonterminal. Every right-hand side is a nonempty word over terminals and nonterminals, and the dependency graph is acyclic. Each symbol $A$ therefore has a uniquely determined terminal expansion $\val(A)$. The word generated by the SLP is $\val(S)$. Its size is
\[
 |G|=\sum_A |\alpha_A|,
\]
and $g(w)$ denotes the minimum size of an SLP generating $w$. A notation such as $A^k$ in a production abbreviates $k$ literal copies of $A$; it contributes $k$ to the size.

\paragraph{The global algorithm.}
\Greedy{} starts with $S\to w$. For a factor $v$ of the current right-hand sides, let $c_G(v)$ be the sum, over these right-hand sides, of the maximum numbers of nonoverlapping occurrences of $v$. A factor is \emph{eligible} if $|v|\geq 2$ and $c_G(v)\geq 2$. Replacing its occurrences by a fresh nonterminal $Y$ and then adding $Y\to v$ reduces the grammar size by
\begin{equation}\label{eq:gain}
 \Delta_G(v)=c_G(v)(|v|-1)-|v|=(c_G(v)-1)|v|-c_G(v).
\end{equation}
Occurrences are chosen from left to right in every old right-hand side. This selects a maximum nonoverlapping set: replacing the first occurrence in any optimal set by the leftmost occurrence cannot obstruct a later occurrence, and the same argument applies inductively to the remaining suffix.

In the definition of the global algorithms in~\cite{charikar}, an eligible factor is maximal if no strictly longer factor has at least its occurrence count. \Greedy{} maximizes~\eqref{eq:gain} among such factors. Equivalently, it maximizes~\eqref{eq:gain} among all eligible factors. Indeed, if lengths $\ell<L$ and counts $2\leq c\leq C$ witness nonmaximality, then
\[
 \bigl(C(L-1)-L\bigr)-\bigl(c(\ell-1)-\ell\bigr)
 \geq (c-1)(L-\ell)>0.
\]
Every maximizer among the finitely many eligible factors is therefore maximal. In particular, any eligible factor supplies a lower bound on the gain of the next chosen factor.

The algorithm stops when no eligible factor remains; its resulting grammar is called terminal. Substitution preserves the generated word and acyclicity. For the latter, a new passage $U\to Y\to V$ in the dependency graph can be replaced by the old edge $U\to V$; a new cycle would thus give an old cycle. To check the word, interpret $Y$ as the old expansion of $v$ and keep all old symbol values. These values satisfy all the new rules. Finally, every eligible gain is nonnegative, every right-hand side remains nonempty, and each round adds one rule. Hence every execution terminates after at most $|w|-1$ rounds. We allow rounds of gain zero; all substitutions forced below have positive gain.

\paragraph{Counting factors.}
We use a form of the usual distinct-factor bound that also allows erasing terminals. Write $\Fac_L(z)$ for the set of length-$L$ factors of a word $z$.

\begin{lemma}\label{lem:capacity}
Let $G$ be an SLP with start symbol $S$. Suppose that each symbol is interpreted as a word over an alphabet $\Gamma$, all rule equations are satisfied, and each terminal is interpreted as either one letter or the empty word. If $z$ is the interpretation of $S$, then, for every $L\geq 2$,
\[
 |\Fac_L(z)|\leq |G|(L-1).
\]
\end{lemma}
\begin{proof}
Consider an occurrence of a length-$L$ factor in the derivation tree for $S$, with leaves interpreted as above. Descend into a child containing the whole occurrence as long as this is possible. This process cannot end at a terminal leaf, whose interpretation has length at most one. At the final node, the occurrence crosses a boundary between consecutive children, with at least one letter on each side. Choose one such boundary and let $k\in\{1,\ldots,L-1\}$ be the number of letters of the occurrence before it. The production, the chosen boundary, and $k$ uniquely determine the factor: it is the length-$k$ suffix before the boundary followed by the length-$(L-k)$ prefix after it. This remains true when some children have empty interpretations. A right-hand side of length $d$ has $d-1$ boundaries. Thus the number of distinct factors is at most
\[
 (L-1)\sum_A (|\alpha_A|-1)\leq (L-1)|G|.\qedhere
\]
\end{proof}

\section{The construction}\label{sec:construction}

The input has two parts. The \emph{targets} are unary runs with carefully chosen lengths. The \emph{auxiliaries} force \Greedy{} to read these lengths in base $64$, one digit at a time. Distinct separators prevent substitutions from crossing between these words.

\subsection{A power-free auxiliary word}\label{sec:auxiliary}

We use the following $19$-uniform morphism on $\{0,1,2\}$, due to Dejean~\cite{dejean}:
\begin{equation}\label{eq:morphism}
 \begin{aligned}
 \mu(0)&=0120212012102120210,\\
 \mu(1)&=1201020120210201021,\\
 \mu(2)&=2012101201021012102.
 \end{aligned}
\end{equation}
A word has period $p\geq 1$ if letters at distance $p$ agree whenever both exist. It is \emph{$(7/4)^+$-free} if none of its factors of length $\ell$ and period $p$ satisfies $4\ell>7p$.

\begin{lemma}\label{lem:morphism}
Every word $\mu^e(0)$, $e\geq 0$, is $(7/4)^+$-free. In each word $\mu^2(c)$, the number of occurrences of any fixed factor of length $1$, $2$, or $3$ is at most $121$, $61$, or $41$, respectively.
\end{lemma}

A self-contained proof, including the finite calculations for this particular morphism, is given in \cref{app:morphism}. The consequence needed in the main argument is a uniform bound on competing substitutions in a padded auxiliary word.

\begin{lemma}\label{lem:aux-gain}
Let $e\geq 2$, $V=\mu^e(0)$, and $M=19^e$. Replace every letter $c$ of $V$ by $X^L c$, where $L\geq 1$ and $X\notin\{0,1,2\}$, obtaining a word $F$. If an eligible factor $v$ of $F$ contains a letter other than $X$, then its gain, computed using its occurrences in $F$, is at most
\[
 \Delta_F(v)\leq \frac{15}{16}M(L+1).
\]
\end{lemma}
\begin{proof}
First, equal length-$m$ factors in $V$ have starting positions at distance at least $4m/3$. If their distance is $d$, the factor spanning both occurrences has length $d+m$ and period $d$. By \cref{lem:morphism}, $4(d+m)\leq 7d$. Consequently, $c$ occurrences of a fixed length-$m$ factor satisfy
\begin{equation}\label{eq:spacing}
 M\geq m+(c-1)\frac{4m}{3}.
\end{equation}

Now let $m\geq 1$ be the number of non-$X$ letters in $v$, and let $c\geq 2$ be its nonoverlapping occurrence count in $F$. Deleting the $X$'s from each occurrence identifies a distinct occurrence of one fixed length-$m$ factor in $V$: the first non-$X$ letter has a fixed offset in $v$ and determines its starting position. Moreover,
\begin{equation}\label{eq:padded-length}
 |v|\leq (m+1)L+m<(m+1)(L+1).
\end{equation}
For $m\geq 4$, \cref{eq:gain,eq:spacing,eq:padded-length} give
\[
 \Delta_F(v)\leq (c-1)(m+1)(L+1)
 \leq \frac34(M-m)\frac{m+1}{m}(L+1)
 \leq \frac{15}{16}M(L+1).
\]
For $m\in\{1,2,3\}$, partition $V$ into its $M/361$ consecutive $\mu^2$-blocks. At most $m-1$ length-$m$ occurrences cross each boundary. The second assertion of \cref{lem:morphism} bounds the internal occurrences in every block by $a_m$, where $(a_1,a_2,a_3)=(121,61,41)$. Thus
\[
 \Delta_F(v)\leq c|v|
 \leq M(L+1)\frac{(m+1)(a_m+m-1)}{361}.
\]
The three numerators are $242$, $186$, and $172$, all less than $361\cdot 15/16$.
\end{proof}

\subsection{Targets, auxiliaries, and separators}\label{sec:input}

Fix $h=2^r$, $r\geq 10$, and set
\[
 B=64,\qquad q=2r,\qquad R=h^2,\qquad H=h^2B^{2h}.
\]
Let $D$ be a binary cyclic de Bruijn word of order $q$, written linearly with length $2^q=h^2$. Thus every binary word of length $q$ occurs at exactly one cyclic starting position, and a length-$q$ factor contained in the linear word $D$ occurs at most once. For definiteness, choose the lexicographically least linear representative among all such words.

Recall briefly why such a word exists. The directed graph with vertices $\{0,1\}^{q-1}$ and edges $\{0,1\}^q$, directed from prefix to suffix of length $q-1$, is strongly connected and has indegree and outdegree two at every vertex. Following unused edges produces a closed walk, since balance prevents stopping at a different vertex. If edges remain, strong connectivity gives a visited vertex with an unused outgoing edge. Splice a new closed walk from that vertex into the old one and repeat. The resulting walk uses every edge once; reading the last letters of its edges cyclically gives the required word.

Apply the code $\varphi(0)=01$, $\varphi(1)=10$ and write
\[
 W=\varphi(D)=C_1C_2\cdots C_h,\qquad |C_s|=2h.
\]
Positions within a chunk are numbered from $1$. For $1\leq s\leq h$ and $h+1\leq i\leq 2h$, define
\begin{equation}\label{eq:target-integers}
 n_{s,i}=\sum_{p=1}^{i} C_s[p]B^{i-p}.
\end{equation}
There are $R=h^2$ such integers. All are positive, since the first two letters of each chunk are $01$ or $10$, and each is less than $B^{2h}$. Their sum is therefore at most $H$.

For $1\leq j\leq h$, put
\begin{equation}\label{eq:aux-parameters}
 e_j=7h+2-3j,\qquad M_j=19^{e_j},\qquad V_j=\mu^{e_j}(0).
\end{equation}
Choose three private terminal letters $b_{j,0},b_{j,1},b_{j,2}$ for each $j$, with all these letters distinct and different from $a$. Let $F_j$ be the image of $V_j$ under
\begin{equation}\label{eq:aux-input}
 c\longmapsto a^{B^j}b_{j,c}.
\end{equation}
The word $w_h$ consists of the targets $a^{n_{s,i}}$, in lexicographic order of $(s,i)$, followed by $F_1,\ldots,F_h$, with a fresh terminal separator between every two successive words. All $h^2+h-1$ separators are distinct and different from the other terminals.

A separator occurs only once in the initial grammar, so it belongs to no repeated factor. Inductively, it remains only in the start rule, and no substitution crosses it. We may therefore speak of the target and auxiliary \emph{fragments} of the start right-hand side. These fragments are not additional rules.

Two scale inequalities will be useful:
\begin{equation}\label{eq:scales}
 \frac{M_j}{M_{j+1}}=19^3=6859>B^2,
 \qquad M_h>H+B+100.
\end{equation}
Indeed, $h\leq 2^h$ gives $H\leq 2^{14h}$, whereas
\[
 M_h=19^{4h+2}>2^{16h+8}\geq 1024\cdot 2^{14h}\geq 1024H.
\]
Since $H\geq 4096$, the second inequality follows.

\subsection{A small comparison grammar}\label{sec:comparison}

\begin{proposition}\label{prop:comparison}
The word $w_h$ satisfies
\[
 h^2+4h\leq g(w_h)\leq 600h^2,
 \qquad 2^h\leq |w_h|\leq 2^{50h}.
\]
\end{proposition}
\begin{proof}
We first construct an SLP of size at most $600h^2$. Within each chunk, successive base-$B$ prefixes satisfy $n'=Bn+b$, where $b\in\{0,1\}$. A nonterminal for a positive prefix can therefore be defined using $B$ copies of the preceding prefix nonterminal and, if $b=1$, one copy of $a$. Before the first $1$ there is no rule; the first positive prefix is represented by $a$ itself. At most $2h$ rules of length at most $B+1=65$ suffice per chunk. All target words are thus available at total cost at most $130h^2$.

Independently, define $P_0=a$ and $P_j\to P_{j-1}^B$ for $1\leq j\leq h$, at cost $64h$. For each $j$ and $c\in\{0,1,2\}$, define
\[
 Q_{j,0,c}\to P_j b_{j,c}.
\]
If $\mu(c)=c_1\cdots c_{19}$, define, for $1\leq u\leq e_j$,
\[
 Q_{j,u,c}\to Q_{j,u-1,c_1}\cdots Q_{j,u-1,c_{19}}.
\]
The symbol $Q_{j,e_j,0}$ generates $F_j$. These rules cost $6h+57\sum_j e_j$, where
\[
 \sum_{j=1}^h e_j=\frac{11h^2+h}{2}.
\]
The start rule lists the $h^2+h$ target and auxiliary symbols, interspersed with their separators, and has length $2h^2+2h-1$. All dependencies in this construction are acyclic. Adding the costs gives
\[
 g(w_h)\leq \frac{891h^2+201h-2}{2}
 \leq 546h^2\leq 600h^2.
\]

Every private auxiliary letter occurs, since each image under $\mu$ contains all three letters. Including $a$ and the separators, the alphabet size is
\[
 1+3h+(h^2+h-1)=h^2+4h.
\]
Each terminal appearing in a generated word must occur literally in some production, proving the lower bound on $g(w_h)$.

Finally, $|w_h|\geq M_h>2^{16h+8}\geq 2^h$. For the upper bound, the total target length is at most $H\leq 2^{14h}$, and the total auxiliary length is at most
\[
 \sum_{j=1}^h M_j(B^j+1)
 \leq 2h\,19^{7h}64^h
 \leq 2^{42h+1}.
\]
There are at most $2h^2\leq 2^{2h+1}$ separators. Each of the three contributions is at most $2^{42h+1}$, so
\[
 |w_h|\leq 3\cdot 2^{42h+1}\leq 2^{42h+3}\leq 2^{50h}.\qedhere
\]
\end{proof}

\section{The substitutions forced by the auxiliaries}\label{sec:phases}

Write $A_0=a$. We show that every execution eventually creates nonterminals $A_1,\ldots,A_h$, in this order, by the substitutions
\begin{equation}\label{eq:phases}
 A_t\to A_{t-1}^{B},\qquad 1\leq t\leq h.
\end{equation}
Other substitutions may occur between them. Equation~\eqref{eq:phases} describes a rule when it is created; later global substitutions may change its right-hand side. Each $A_t$ always expands to $a^{B^t}$, but in this section it is a single literal symbol, called a \emph{primary symbol}.

\subsection{An invariant for the whole grammar}\label{sec:invariant}

Immediately before phase $t$, let $X=A_{t-1}$. We maintain the following invariant, which accounts for every fragment of the start rule and every other right-hand side.
\begin{enumerate}[label=(\roman*)]
\item\label{inv:targets}
The fragment for a target integer $n$ is
\begin{equation}\label{eq:target-invariant}
 X^{\lfloor n/B^{t-1}\rfloor}
 A_{t-2}^{d_{t-2}}\cdots A_0^{d_0},
 \qquad d_u\in\{0,1\},
\end{equation}
where $d_u$ are the lower base-$B$ digits of $n$. The lower sequence is empty when $t=1$. There are no additional target rules, and the total length of the target fragments is at most $H$.
\item\label{inv:unfinished}
For $j\geq t$, auxiliary $j$ is unfinished: its fragment is the image of $V_j$ under
\[
 c\longmapsto X^{L_j}b_{j,c},\qquad L_j=B^{j-t+1}.
\]
Its markers appear nowhere else, and it has no additional rules.
\item\label{inv:finished}
For $j<t$, the fragment of finished auxiliary $j$ and its additional rules use only $A_j$, its three markers, and nonterminals private to that auxiliary. None of these fragments or right-hand sides contains $A_jA_j$.
\item\label{inv:power}
For each $0\leq u\leq t-2$, the defining rule for $A_{u+1}$ and its additional rules form a power layer. Every right-hand side in this layer uses only $A_u$ and nonterminals private to this layer.
\end{enumerate}
Here each nonprimary, nonstart nonterminal is assigned to exactly one finished auxiliary or power layer. Its defining rule and all its occurrences belong to that class. The private alphabets of different classes are disjoint. A primary symbol may occur in several classes, but its defining rule belongs to the specified power layer. In particular, no power-layer right-hand side contains the active symbol $X$. Separators occur only between the start fragments.

The invariant holds initially for $t=1$: there are only the original target and auxiliary fragments. We prove that every chosen factor other than $X^B$ is confined to one already processed class and preserves the invariant. The factor $X^B$ advances it to phase $t+1$.

\subsection{All factors affecting an unprocessed fragment lose}\label{sec:competitors}

The factor $X^B$ has $M_t$ nonoverlapping occurrences in auxiliary $t$ alone. It is eligible, and its global gain $D_B$ satisfies
\begin{equation}\label{eq:benchmark}
 D_B\geq M_t(B-1)-B>H.
\end{equation}
We compare it with every other eligible factor that occurs in a target or an unfinished auxiliary.

\proofheading{Active unary factors}
For $d\geq 2$, occurrences of $X^d$ lie only in the target prefixes and unfinished auxiliaries. Power layers contain no $X$; finished auxiliaries with index below $t-1$ use a different primary symbol; and finished auxiliary $t-1$ has no adjacent $X$'s.

Suppose first that $2\leq d<B$, and write $B=ud+v$ with $0\leq v<d$. In each length-$B$ run of auxiliary $t$, the replacement benefit before paying for the new rule differs by
\[
 (B-1)-u(d-1)=u+v-1\geq 1
\]
in favor of $X^B$. Indeed, either $u\geq 2$, or $u=1$ and $v=B-d\geq 1$. In any other unfinished run, its length $L$ is divisible by $B$, and
\[
 \lfloor L/d\rfloor(d-1)\leq L(1-1/d)<L(1-1/B).
\]
Such a run also favors $X^B$. The targets can favor $X^d$ by at most their total length $H$, and the difference in new-rule costs is at worst $-B$. Hence
\begin{equation}\label{eq:short-unary}
 D_B-\Delta_G(X^d)\geq M_t-H-B>0.
\end{equation}

If $d>B$, auxiliary $t$ has no occurrence of $X^d$. By \eqref{eq:scales}, the total number of active symbols in the later auxiliaries is
\begin{equation}\label{eq:later-mass}
 \sum_{j=t+1}^h M_j B^{j-t+1}
 \leq M_t B\sum_{s=1}^{h-t}B^{-s}
 <\frac{B}{B-1}M_t<2M_t.
\end{equation}
The benefit in these auxiliaries and the targets is therefore less than $2M_t+H$. On the other hand,
\[
 D_B-(2M_t+H)\geq M_t(B-3)-B-H>0.
\]
This excludes every active unary factor except $X^B$.

\proofheading{Factors containing an unfinished auxiliary's marker}
Every occurrence of such a factor lies in the same auxiliary $j\geq t$, because that marker occurs nowhere else. By \cref{lem:aux-gain}, its gain is at most $\frac{15}{16}M_j(L_j+1)$. The occurrences of $X^B$ in this auxiliary alone give gain at least $M_jL_j(1-1/B)-B$. Thus
\begin{align*}
 D_B-\Delta_G(v)
 &\geq M_j\left(\frac{3L_j}{64}-\frac{15}{16}\right)-64\\
 &\geq \frac{33}{16}M_j-64>0,
\end{align*}
since $L_j\geq B=64$ and $M_j\geq M_h>H+B+100$.

\proofheading{Other factors in targets}
Every target consists of primary symbols. A power of a lower primary symbol cannot occur: after the active prefix, the primary indices in~\eqref{eq:target-invariant} decrease strictly. Any remaining factor of length at least two therefore contains two distinct primary symbols. Only target fragments can contain such a factor, since each other class uses at most one primary symbol. Its selected occurrences cover at most $H$ symbols, so its gain is at most $H<D_B$.

These comparisons are strict. Consequently, every chosen factor other than $X^B$ avoids all targets and all unfinished auxiliaries.

\subsection{Intervening substitutions and progress}\label{sec:cleanup}

Consider a chosen factor $v\ne X^B$. It has no separator and avoids all unprocessed fragments. If it contains a private nonterminal or marker, that symbol identifies one finished auxiliary or power layer, and all occurrences of $v$ lie in that class. If it contains only primary symbols, it cannot contain two distinct ones, since only targets could contain such a factor. It is therefore a power $A_u^d$, $d\geq 2$. The active case was excluded above. The finished auxiliary using $A_u$ contains no adjacent $A_u$'s, so all occurrences lie in power layer $u$.

In either case, assign the fresh nonterminal to this one class. All replacements and the new rule stay in the class. Its permitted alphabet is preserved. In a finished auxiliary, replacing a nonempty factor by a fresh symbol different from $A_j$ creates no new adjacent pair $A_jA_j$, and the new right-hand side is an old factor, so it too contains no such pair. Thus every part of the invariant is preserved.

As long as phase $t$ has not occurred, $X^B$ remains eligible with positive gain~\eqref{eq:benchmark}. The execution cannot terminate, and it cannot perform infinitely many intervening substitutions, by the termination bound in \cref{sec:preliminaries}. It must therefore eventually choose $X^B$.

At that step, each unfinished auxiliary run has length divisible by $B$, so left-to-right replacement changes it exactly to $A_t^{L_j/B}$. Auxiliary $t$ becomes an alternating word of $A_t$'s and private markers and hence satisfies the finished-auxiliary condition. For a target prefix, write
\[
 \left\lfloor\frac{n}{B^{t-1}}\right\rfloor
 =B\left\lfloor\frac{n}{B^t}\right\rfloor+d_{t-1},
 \qquad d_{t-1}\in\{0,1\}.
\]
Left-to-right replacement gives the new prefix $A_t^{\lfloor n/B^t\rfloor}A_{t-1}^{d_{t-1}}$, followed by the unchanged lower digits. The new rule $A_t\to A_{t-1}^B$ begins power layer $t-1$. No other class contains $X^B$. This proves the invariant for the next phase.

The preceding induction proves the following statement.

\begin{proposition}\label{prop:forced}
Every execution of \Greedy{} on $w_h$ performs all substitutions in~\eqref{eq:phases}, in order. Immediately after phase $h$, the target fragment for $(s,i)$ is
\begin{equation}\label{eq:target-after}
 A_h^{\lfloor n_{s,i}/B^h\rfloor}U_{s,i},
 \qquad
 U_{s,i}=A_{h-1}^{C_s[i-h+1]}A_{h-2}^{C_s[i-h+2]}
 \cdots A_0^{C_s[i]}.
\end{equation}
\end{proposition}

\section{A lower bound on every final grammar}\label{sec:output}

The words $U_{s,i}$ contain many distinct factors. Later substitutions may change their literal representation, including the definitions of primary symbols. We first show that this does not invalidate a factor-count argument.

\subsection{Keeping the primary symbols as letters}

\begin{lemma}\label{lem:freeze}
Let $G_h$ be any terminal grammar reached after phase $h$. There is a word $z$ containing every $U_{s,i}$ as a factor such that, for every $L\geq 2$,
\[
 |\Fac_L(z)|\leq |G_h|(L-1).
\]
\end{lemma}
\begin{proof}
Immediately after phase $h$, regard $A_0,\ldots,A_h$ as distinct terminal letters: remove the defining rules for $A_1,\ldots,A_h$ and evaluate the remaining acyclic grammar with these symbols held fixed. This assigns an interpretation to every symbol and satisfies every rule except the removed ones. Its start word contains all the fragments in~\eqref{eq:target-after}.

Maintain these interpretations through subsequent substitutions. If a fresh symbol $Y$ replaces $v$, assign to $Y$ the concatenation of the current interpretations of the symbols of $v$, leaving all old interpretations unchanged. Every retained rule equation stays true, and the new equation $Y\to v$ is true by definition. This argument also applies when some occurrences of $v$ lie in the removed primary definitions: no equation for such a definition is required. Thus the interpreted start word is unchanged throughout the rest of the execution.

In the final grammar, again delete the $h$ primary definitions and regard their left-hand sides as terminal letters. This is an acyclic grammar of size at most $|G_h|$ whose rule equations are satisfied by the maintained interpretations. Now erase the letter $A_h$ from all interpreted words and keep every other terminal letter. The interpretation of each terminal has length at most one. The resulting start word $z$ still contains every $U_{s,i}$, since none of them contains $A_h$. Apply \cref{lem:capacity} to this grammar and interpretation.
\end{proof}

\subsection{Counting the factors in the targets}

Set
\begin{equation}\label{eq:factor-length}
 L=2q+2=4r+2.
\end{equation}
For $r\geq 10$, we have $L\leq h/4$: it holds at $r=10$, and increasing $r$ adds $4$ to the left side and doubles the right side.

\begin{lemma}\label{lem:binary-multiplicity}
Every factor of $W=\varphi(D)$ of length at least $L$ occurs at most twice.
\end{lemma}
\begin{proof}
Fix the parity of an occurrence's starting position. After deleting at most one letter at each end, the occurrence contains at least $q$ whole codewords of $\varphi$. Equal factors with this starting parity therefore identify equal length-$q$ factors of $D$. Such a factor of $D$ has only one starting position. The corresponding start in $W$ is then fixed as well. There is at most one occurrence of each parity.
\end{proof}

\begin{lemma}\label{lem:target-factors}
Across the $h^2$ words $U_{s,i}$, there are at least $h^3/8$ distinct factors of length $L$.
\end{lemma}
\begin{proof}
Each $U_{s,i}$ records the $1$'s in a length-$h$ binary interval of $W$. Such an interval contains at least $h/2-1$ complete codewords of $\varphi$, each with one $1$. Hence
\[
 |U_{s,i}|\geq h/2-1.
\]
Since $L\leq h/4$, this gives at least
\[
 |U_{s,i}|-L+1\geq h/2-L\geq h/4
\]
length-$L$ factor occurrences in each $U_{s,i}$, for a total of at least $h^3/4$ occurrences.

It remains to show that any fixed factor is counted at most twice. Write it as
\[
 A_{j_1}A_{j_2}\cdots A_{j_L},\qquad
 h-1\geq j_1>j_2>\cdots>j_L\geq 0.
\]
These indices determine a binary word of length $j_1-j_L+1\geq L$: it has $1$'s exactly at offsets $0,j_1-j_2,\ldots,j_1-j_L$, and $0$'s elsewhere. For an occurrence in $U_{s,i}$, this word occurs in $W$ starting at position
\begin{equation}\label{eq:position}
 2h(s-1)+i-j_1.
\end{equation}
There can be no additional $1$ between the recorded ones, since that would give an intervening primary symbol in the factor. By \cref{lem:binary-multiplicity}, the binary word has at most two starting positions in $W$. Each starting position belongs to a unique chunk, determining $s$, and then~\eqref{eq:position} determines $i$. Within $U_{s,i}$, each primary index occurs at most once, so the occurrence is unique. The multiplicity is therefore at most two. Dividing the total occurrence count by two proves the claim.
\end{proof}

\begin{proof}[Proof of \cref{thm:main}]
The construction and \cref{prop:comparison} give the alphabet, minimum-grammar, and input-length bounds. By \cref{prop:forced}, every execution reaches the state described in~\eqref{eq:target-after}. For any final grammar $G_h$, \cref{lem:freeze,lem:target-factors} give
\[
 \frac{h^3}{8}\leq |\Fac_L(z)|\leq |G_h|(L-1).
\]
Since $L-1=4r+1$, this proves the required output-size bound. Dividing by $g(w_h)\leq 600h^2$ yields
\[
 \frac{|G_h|}{g(w_h)}\geq \frac{h}{4800(4r+1)}
 \geq \frac{h}{24000r}.
\]
Finally, $n_h\leq 2^{50h}$ gives $h\geq (\log n_h)/50$, and $n_h\geq 2^h$ gives $r=\log h\leq\log\log n_h$. Substitution proves the displayed ratio bound in the theorem. As $r$ tends to infinity, the input lengths are unbounded and $h/(4r+1)$ tends to infinity. This establishes a nonconstant lower bound on an infinite family.
\end{proof}

\appendix
\section{The finite morphism calculation}\label{app:morphism}

We prove \cref{lem:morphism} directly for~\eqref{eq:morphism}. The proof has two parts: a synchronization argument reduces power-freeness to images of words of length at most three, and a separate short-factor count gives the constants used in \cref{lem:aux-gain}. Positions in this appendix are numbered from $0$; intervals are half-open.

\subsection{Synchronization}

The morphism commutes with the cyclic permutation $0\mapsto1\mapsto2\mapsto0$. Each $\mu(c)$ starts and ends with $c$. The following table lists all occurrences of the length-seven prefixes and suffixes of the three images in $\mu(a)\mu(b)$, for $a\ne b$, whose starting offsets lie in $\{0,\ldots,18\}$. The symbol $\varnothing$ means that there is no such occurrence.

\begin{table}[htbp]
\centering
\begin{tabular}{@{}lrrrrrr@{}}
\toprule
Factor & $01$ & $02$ & $10$ & $12$ & $20$ & $21$\\
\midrule
0120212 & 0 & 0 & $\varnothing$ & $\varnothing$ & $\varnothing$ & $\varnothing$ \\
2120210 & 12 & 12 & $\varnothing$ & $\varnothing$ & $\varnothing$ & $\varnothing$ \\
1201020 & $\varnothing$ & $\varnothing$ & 0 & 0 & $\varnothing$ & $\varnothing$ \\
0201021 & $\varnothing$ & $\varnothing$ & 12 & 12 & $\varnothing$ & $\varnothing$ \\
2012101 & $\varnothing$ & $\varnothing$ & $\varnothing$ & $\varnothing$ & 0 & 0 \\
1012102 & $\varnothing$ & $\varnothing$ & $\varnothing$ & $\varnothing$ & 12 & 12 \\
\bottomrule
\end{tabular}
\caption{Synchronization in two consecutive images. The column $ab$ refers to $\mu(a)\mu(b)$; entries are starting offsets.}\label{tab:sync}
\end{table}

\begin{lemma}\label{lem:sync}
If $x$ has no equal adjacent letters, then two equal length-$18$ factors in $\mu(x)$ start at positions congruent modulo $19$.
\end{lemma}
\begin{proof}
Let $t$ be the starting offset of the first occurrence in its $\mu$-block. If $t=0$, its first seven letters are the block prefix. If $1\leq t\leq12$, the occurrence contains the last seven letters of that block. If $13\leq t\leq18$, it contains the first seven letters of the next block. By \cref{tab:sync}, these length-seven words can start only at offset $0$ for a prefix or offset $12$ for a suffix. In the equal occurrence, the same seven letters at the same relative position therefore fix the starting residue modulo $19$. For an occurrence in the final block, one may append a different letter to $x$ when applying the two-block table; this does not change any existing factor.
\end{proof}

\subsection{Preservation of power-freeness}

For a word $y$ and $p\geq1$, let $R_y(p)$ be the largest number of consecutive true equalities
\[
 y[j]=y[j+p]
\]
as $j$ runs through $0,\ldots,|y|-p-1$; put $R_y(p)=0$ if $p\geq|y|$. A factor of period $p$ and length greater than $7p/4$ exists exactly when
\begin{equation}\label{eq:run-test}
 4R_y(p)>3p.
\end{equation}
Indeed, a run of $R$ equalities gives a factor of length $p+R$ and period $p$, and every factor of length greater than $p$ and period $p$ supplies such a run.

Every nonempty $(7/4)^+$-free word of length at most three is a cyclic renaming of one of
\[
 \mathcal X=\{0,01,02,010,012,020,021\}.
\]
This follows because equal adjacent letters are forbidden, and the first letter can be renamed to $0$. For these seven words, define
\[
 R(p)=\max_{x\in\mathcal X}R_{\mu(x)}(p).
\]
Their images have length at most $57$. Only $p<4\cdot57/7$, hence $1\leq p\leq32$, could violate power-freeness. \Cref{tab:runs} gives all the required values and shows $R(p)\leq\lfloor3p/4\rfloor$ in every case. The entries can be obtained by scanning the explicitly given words: start with $b_{-1}=0$ and set $b_j=b_{j-1}+1$ when $y[j]=y[j+p]$, and $b_j=0$ otherwise; then take the maximum $b_j$.

\begin{table}[htbp]
\centering
\begin{tabular}{@{}rrr@{\qquad\qquad}rrr@{}}
\toprule
$p$ & $\lfloor3p/4\rfloor$ & $R(p)$ & $p$ & $\lfloor3p/4\rfloor$ & $R(p)$\\
\midrule
1 & 0 & 0 & 17 & 12 & 3 \\
2 & 1 & 1 & 18 & 13 & 7 \\
3 & 2 & 2 & 19 & 14 & 0 \\
4 & 3 & 3 & 20 & 15 & 7 \\
5 & 3 & 2 & 21 & 15 & 3 \\
6 & 4 & 3 & 22 & 16 & 2 \\
7 & 5 & 4 & 23 & 17 & 2 \\
8 & 6 & 5 & 24 & 18 & 6 \\
9 & 6 & 1 & 25 & 18 & 9 \\
10 & 7 & 2 & 26 & 19 & 5 \\
11 & 8 & 3 & 27 & 20 & 3 \\
12 & 9 & 5 & 28 & 21 & 2 \\
13 & 9 & 6 & 29 & 21 & 2 \\
14 & 10 & 4 & 30 & 22 & 5 \\
15 & 11 & 2 & 31 & 23 & 4 \\
16 & 12 & 2 & 32 & 24 & 9 \\
\bottomrule
\end{tabular}
\caption{All period checks for images of the seven short representatives.}\label{tab:runs}
\end{table}

We now show that $\mu$ preserves $(7/4)^+$-freeness. Suppose otherwise, and choose a shortest $(7/4)^+$-free word $x$ such that $\mu(x)$ has a factor $\mu(x)[a:b)$ of length $\ell=b-a$ and period $p$ with $4\ell>7p$. The preceding checks show that $k=|x|\geq4$. By minimality, the factor touches both the first and last $\mu$-blocks, so
\begin{equation}\label{eq:bad-bounds}
 0\leq a\leq18,\qquad 19(k-1)+1\leq b\leq19k,
 \qquad \ell\geq19(k-2)+2\geq40.
\end{equation}
The period gives equal length-$(\ell-p)$ factors at positions $a$ and $a+p$. Since
\[
 \ell-p>\frac{3\ell}{7}\geq\frac{120}{7}>17,
\]
these equal factors have length at least $18$. The word $x$ has no equal adjacent letters, so \cref{lem:sync} implies $p=19d$ for an integer $d\geq1$. Moreover,
\[
 d<\frac{4k}{7}<k-1,
\]
and hence $d\leq k-2$.

We recover a period $d$ in $x$. The equality at position $18$ is available, because $a\leq18$ and
\[
 b-p\geq19(k-d-1)+1\geq20.
\]
Comparing the last letters of blocks $0$ and $d$ gives $x[0]=x[d]$. For $1\leq i\leq k-d-1$, position $19i$ lies in $[a,b-p)$, so comparison of the first letters of blocks $i$ and $i+d$ gives $x[i]=x[i+d]$. Thus $x$ has period $d$. But $7p<4\ell\leq76k$ implies $7d<4k$, contradicting the power-freeness of $x$.

The one-letter word $0$ is $(7/4)^+$-free. Induction now proves that every $\mu^e(0)$ is $(7/4)^+$-free, as asserted in \cref{lem:morphism}.

\subsection{Counts of factors of lengths one, two, and three}

Let $N_v(y)$ be the number of all occurrences of $v$ in $y$, including overlapping occurrences. Since the letter multiplicities in $\mu(0)$ are $(6,6,7)$,
\begin{equation}\label{eq:count-blocks}
 N_v(\mu^2(0))=6N_v(\mu(0))+6N_v(\mu(1))+7N_v(\mu(2))+E_v,
\end{equation}
where $E_v$ counts occurrences crossing a block boundary. \Cref{tab:counts} gives each term for every word of length at most three without equal adjacent letters. All other words of these lengths have count zero by power-freeness.

For clarity, the boundary counts require only the six adjacent-pair counts in $\mu(0)$:
\[
 (P_{01},P_{02},P_{10},P_{12},P_{20},P_{21})=(2,3,2,4,3,4).
\]
Each $\mu(x)$ begins with $x(x+1)$ and ends with $(x+1)x$, with addition modulo three. For a one-letter word, $E_v=0$; for a two-letter word, $E_v=P_v$; and for a three-letter word,
\begin{equation}\label{eq:boundary-counts}
 E_v=\sum_{x\ne y}P_{xy}
 \bigl([v=(x+1)xy]+[v=xy(y+1)]\bigr),
\end{equation}
where a bracket is $1$ if its equality holds and $0$ otherwise. The internal counts come from sliding $v$ along the three length-$19$ images in~\eqref{eq:morphism}; \cref{eq:count-blocks,eq:boundary-counts} then give the last two columns.

\begin{table}[htbp]
\centering
\begin{tabular}{@{}lrrrrr@{}}
\toprule
$v$ & $N_v(\mu(0))$ & $N_v(\mu(1))$ & $N_v(\mu(2))$ & $E_v$ & $N_v(\mu^2(0))$\\
\midrule
0 & 6 & 7 & 6 & 0 & 120 \\
1 & 6 & 6 & 7 & 0 & 121 \\
2 & 7 & 6 & 6 & 0 & 120 \\
01 & 2 & 3 & 4 & 2 & 60 \\
02 & 3 & 4 & 2 & 3 & 59 \\
10 & 2 & 3 & 4 & 2 & 60 \\
12 & 4 & 2 & 3 & 4 & 61 \\
20 & 3 & 4 & 2 & 3 & 59 \\
21 & 4 & 2 & 3 & 4 & 61 \\
010 & 0 & 2 & 1 & 0 & 19 \\
012 & 2 & 1 & 3 & 2 & 41 \\
020 & 0 & 2 & 0 & 6 & 18 \\
021 & 3 & 2 & 1 & 4 & 41 \\
101 & 0 & 0 & 2 & 4 & 18 \\
102 & 1 & 3 & 2 & 3 & 41 \\
120 & 3 & 2 & 1 & 4 & 41 \\
121 & 1 & 0 & 2 & 0 & 20 \\
201 & 1 & 3 & 2 & 3 & 41 \\
202 & 2 & 1 & 0 & 0 & 18 \\
210 & 2 & 1 & 3 & 2 & 41 \\
212 & 2 & 0 & 0 & 8 & 20 \\
\bottomrule
\end{tabular}
\caption{The short-factor counts used in \cref{lem:aux-gain}.}\label{tab:counts}
\end{table}

The maximum counts in the last column, grouped by factor length, are $121$, $61$, and $41$. As $\mu$ commutes with cyclic renaming, the counts in $\mu^2(1)$ and $\mu^2(2)$ are permutations of these counts. This proves the second assertion of \cref{lem:morphism} and completes its proof.

\FloatBarrier
\section{Lean formalisation}\label{app:lean}

The formalisation uses Lean~4.19.0 and a pinned version of Mathlib. It defines straight-line grammars, their size and semantics, and global \Greedy{} with left-to-right nonoverlapping occurrence replacement. It also proves equivalence with the maximal-factor formulation in~\cite{charikar}. The verified argument includes the existence of the cyclic de Bruijn seed, preservation of power-freeness by the morphism and its finite counting certificates, the forced substitution phases with arbitrary intervening cleanup, the final distinct-factor count, and the construction of a small comparison grammar. These components yield the quantitative lower bound and its logarithmic form. Complete executions are proved to exist, and the final statements quantify over every terminal execution on the selected inputs, including arbitrary maximizing ties. Verification builds the proof and audits the transitive axiom dependencies of all project theorems. No admitted proofs or additional axioms are used; the audit permits only Lean's standard propositional extensionality, classical choice, and quotient soundness axioms. The repository~\cite{lean-proof} includes a correspondence between the paper and the main formal declarations.

\section*{Acknowledgement}GPT-6 Astra was used in developing the proof and preparing the manuscript.

\end{document}